\documentclass[10pt,a4paper,english]{article}

\usepackage[english]{babel}
\usepackage[latin1]{inputenc} 
\usepackage[T1]{fontenc}
\usepackage{amscd, amssymb, mathrsfs,amsfonts, amsthm} 
\usepackage{mathtools}
\usepackage{ellipsis} 
\usepackage{multirow} 
\usepackage[shortcuts]{extdash} 
\usepackage[version=latest]{pgf}
\usepackage{bbm} 

\usepackage{parskip} 
\usepackage[square,numbers]{natbib} 
\usepackage{tikz}
\usetikzlibrary{arrows} 
\usepackage{hyperref} 

\newcommand{\cross}{\mathop{\backslash\mathllap{/}}}
\newcommand{\dbars}{\mathop{||}}

\let\oldsqrt\sqrt
\def\sqrt{\mathpalette\DHLhksqrt}
\def\DHLhksqrt#1#2{%
\setbox0=\hbox{$#1\oldsqrt{#2\,}$}\dimen0=\ht0
\advance\dimen0-0.2\ht0
\setbox2=\hbox{\vrule height\ht0 depth -\dimen0}%
{\box0\lower0.4pt\box2}}

\newcommand{\Ecal}   {{\mathcal E }} 
\newcommand{\Vcal}   {{\mathcal V }} 
 \newcommand{\com}{\textnormal{c}} 
\newcommand{\Tfin}{{\mathop{T_{\operatorname{fin}}}}}
\newcommand{\Gfin}{{\mathop{G_{\operatorname{fin}}}}}

\newcommand{\PGW}{\mathbb{P}_{\operatorname{GW}}}
\newcommand{\EGW}{\mathbb{E}_{\operatorname{GW}}}
\newcommand{\Pqu}{\mathbf{P}^{\theta,X}_{{n}}}
\newcommand{\GW}{Galton-Watson }

\numberwithin{equation}{section}
\newtheorem{thrm}{Theorem}
\newtheorem{ex}[thrm]{Example}
\newtheorem{lem}[thrm]{Lemma}
\newtheorem{cor}[thrm]{Corollary}
\newtheorem{prop}[thrm]{Proposition}

\begin{document}

\makeatletter
\renewcommand\section{\@startsection {section}{1}{\z@}%
	{-0.5ex \@plus -2ex \@minus -.2ex}%
	{0.3ex \@plus.2ex}%
	{\normalfont\Large\bfseries}}
\renewcommand\subsection{\@startsection {subsection}{1}{\z@}%
	{-0.5ex \@plus -2ex \@minus -.2ex}%
	{0.3ex \@plus.2ex}%
	{\normalfont\large\bfseries}}
\makeatother

\title{Phase transition for loop representations of Quantum spin systems on trees}
\author{Volker Betz\footnote{betz@mathematik.tu-darmstadt.de}, Johannes Ehlert\footnote{ehlert@mathematik.tu-darmstadt.de}, Benjamin Lees\footnote{benjaminlees90@gmail.com}}
\date{Technische Universit\"at Darmstadt, Germany}
\maketitle

\begin{abstract}
\noindent
We consider a model of random loops 
on \GW trees with an offspring distribution with high expectation. 
We give the configurations a weighting of $\theta^{\#\text{loops}}$.  
For many $\theta>1$ these models are equivalent to certain quantum spin 
systems for various choices of the system parameters. We find conditions on the offspring distribution that guarantee 
the occurrence of a phase transition from finite to infinite loops for the \GW tree. 
\end{abstract}
\section{Introduction}

Loop models are percolation type probabilistic models with intimate 
connections to the correlation functions of certain quantum spin systems. 
To describe them, let $G = (\Vcal,\Ecal)$ be a graph, and for each edge $e \in \Ecal$, 
let $X_e$ be a random variable that takes values in the set of 
finite collections of points (called `links') 
inside an interval $[0,\beta]$. 
The points may be marked, the most important case being that 
there are two different types of points, called crosses and bars. 
Given a configuration $(X_e)_{e \in \Ecal}$, 
a loop configuration is constructed in the following way: each vertex $v$
is assigned a copy of the interval $[0,\beta)_{per}$ (with end points identified), and is then wired to other 
vertices by laying wires' that cross to a neighbouring vertex at 
those places where $X_e$, for an edge $e$ that is incident to 
$v$, has a point. If that point is a bar, the wire is in addition continued 
in the opposite direction on the new edge, otherwise in the 
same direction. Figure \ref{fig:loopexample} gives an illustration of such 
a loop configuration. It is easy to see that this prescription indeed 
results in a configuration of disjoint loops, but that on the other hand 
a vertex can be contained in more than one loop, and that a 
single loop may visit a vertex several times. 
A vertex $v$ is said to be connected to a vertex $v'$ if they share a loop. By thinking of the wiring one may also interpret this as the wiring conducting electricity from $v$ to $v'$. The basic question is about the existence of 
percolation in this sense, i.e. the probability of transferring 
electricity to infinity. 

In the simplest models, the joint probability law of the $(X_e)_{e \in \Ecal}$ 
is a product of independent laws for each $e \in \Ecal$. 
A relevant choice is a pair of independent 
Poisson point processes for crosses and bars, respectively. 
While it is clear that a necessary condition for 
loop percolation is the existence of an infinite cluster of 
edges carrying at least one point, 
the fact that disjoint loops can share an edge makes this 
condition far from sufficient. Except in cases where reflection positivity 
can be applied (see below), very little is known about the existence of infinite 
loops. In the case of random interchange it was shown by Schramm \cite{S} that infinite loops occur on the complete graph. The reader is encouraged to consult the recent review of Ueltschi \cite{U2} and references therein for a more complete overview of current results in this direction.

Loop models that correspond to quantum systems are more complicated: 
they use the independent distribution of the $(X_e)$ as a 
reference measure, but change it with an energy which (in finite volume) 
is proportional to the total number of {\em loops} in the configuration.
We write the corresponding Boltzmann factor in the form
$\theta^{\#\text{loops}}$, $\theta\geq 1$.
Then an infinite volume limit has to be taken. 
Relevant quantum systems include 
the spin-$\tfrac{1}{2}$ Heisenberg ferromagnet, the quantum XY model and a 
spin-1 nematic model. The Heisenberg ferromagnet has been 
investigated using a loop model that is random interchange with 
$\theta = 2$, see \cite{T}. This 
representation was extended by Ueltschi \cite{U} to a family of models 
that, in spin-$\tfrac{1}{2}$, interpolate between the Heisenberg ferro- 
and antiferromagnet and the quantum XY model. The case of the 
antiferromagnet had previously been investigated by Aizenman and 
Nachtergaele \cite{A-N} and the two models agree in this case. 
In all of these examples, one expects that, in the case where the 
reference measure of the $X_e$ is a standard Poisson point process, 
a phase transition occurs: below a critical interval length 
(`temperature') $\beta_c$, loops are finite with probability one, but 
above $\beta_c$ infinite loops appear. This result has
been proved in the case of a cubic lattice of sufficiently high dimensions 
for $\theta$ an integer \cite{U}. The proof relies on the method of
reflection positivity and infra-red bounds which are actually properties 
of the related quantum spin system, and are only available in highly 
symmetric cases such as $\mathbb Z^d$. For the loop models corresponding 
to the quantum XY model and the Heisenberg antiferromagnet (as well as 
interpolations between the two models) this result corresponds to results 
in the papers of Dyson, Lieb and Simon \cite{D-L-S} and Kennedy, Lieb and 
Shastry \cite{K-L-S} for the spin-systems.

In \cite{B-U}, Bj\"ornberg and Ueltschi investigate the loop model 
for $\theta = 1$ (independent $X_e$) in the case where $G$ is a 
$d$-regular tree, and where $d$ is large. The advantage of this setting is 
that $G$ itself has no loops, so the only complications stem from 
doubly occupied edges. By making $d$ large and $\beta = O(1/d)$, 
they can show that in this case the loop percolation threshold 
corresponds, to first order
in $1/d$, to the percolation threshold for occupied edges. What is more, 
they are able to analyse situations where some doubly occupied 
edges are present, and thus get upper and lower bounds for the loop 
percolation threshold that differ from those of edge percolation by a 
term of order $1/d$ and are sharp up to terms of order $1/d^2$. In addition,
these bounds depend on the intensity of crosses and bars, respectively. The
case of $\theta=1$ with only crosses corresponds to random interchange and 
has been previously studied by Angel \cite{A} and Hammond \cite{H1,H2}.
In a very recent preprint \cite{B-U3}, they extend these results to the case 
where $\theta \neq 1$. The justification for studying a tree is that for very
high space dimensions, the difference between a $d$-regular tree and 
$\mathbb Z^d$ in terms of percolation questions should be small. 

In the present paper, we consider the loop model on random trees, more 
precisely on Galton-Watson trees which are strongly supercritical. The idea 
is that if we are in $\mathbb Z^D$ for very high space dimension $D$, 
we can first delete a number of edges that will be unoccupied anyway, 
and have a graph that is approximately a tree, and which on average has 
$1 \ll d \ll D$ children per vertex. We then assume that it 
actually is a tree, but it is still random and certainly not $d$-regular, 
and for mathematical convenience we take it to be a \GW tree. While none of 
these assumptions are strictly true, they are a slightly better 
approximation of the truth than the regular tree assumed in the works of 
Bj\"ornberg and Ueltschi. We investigate this model to first oder in $1/d$, 
and find that up to this order, the loop percolation threshold is again 
equal to the occupied edge percolation threshold on the tree. While this 
result may seem obvious, the presence of the factor 
$\theta^{\#\text{loops}}$ and the fact that the tree may have vertices with 
degree significantly larger than $d+1$ makes the proof non-trivial. 
It is based on ideas from \cite{B-U} and estimates on the effects of links on $\theta^{\#\text{loops}}$. 

\section{Definition and main result}

For a graph $G=(\mathcal V(G),\mathcal E(G))$ fix a parameter $u \in [0,1]$.
We construct, independently on each edge $e \in \mathcal E(G)$, two Poisson processes $N^{e,\cross}$ and $N^{e,\dbars}$ over the time interval $[0,\beta)$ with intensity $ u$ and $(1-u)$, respectively. 
Furthermore, for any finite subgraph, $\Gfin$, of $G$, we denote the joint distribution of $N^{e,*}$ with $e \in \mathcal E(\Gfin)$ and $* \in \{\cross,\dbars\}$ by $\rho_{\Gfin}$. 
For simplicity we write $N^e:= N^{e,\cross}+N^{e,\dbars}$ and say that there is a \textit{link} on an edge $e \in \mathcal E$ at time $t\in[0,\beta)$ iff $N^e$ has a jump at time $t$. Finally, the expectation with respect to $\rho_{\Gfin}$ will be denoted by $\mathbb E_{\Gfin}$.
\begin{figure}[t]
	\includegraphics[width=10cm,height=5cm]{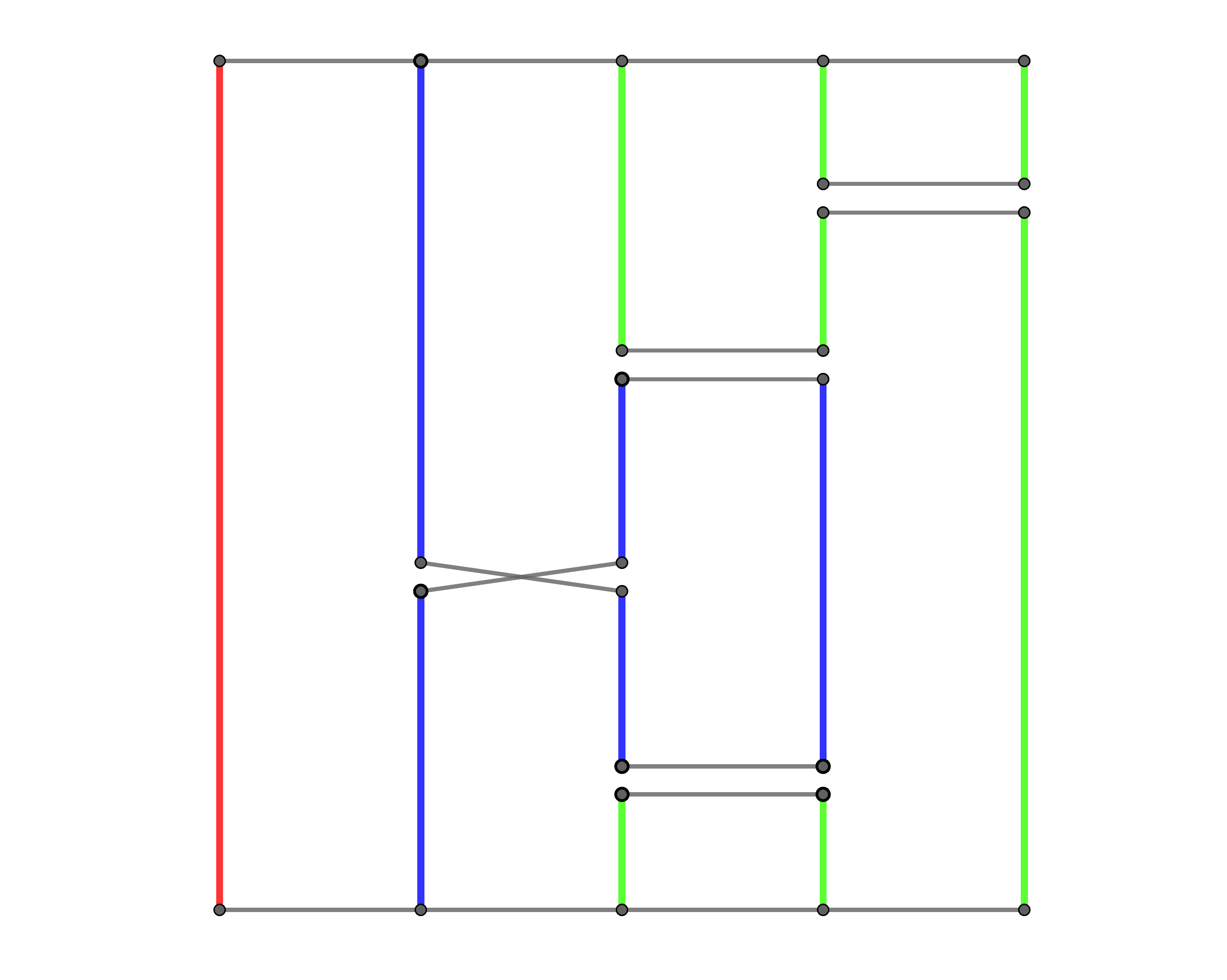}
	\centering
	\caption{A simple example of a realisation, $\omega$, with three loops coloured.}
	\label{fig:loopexample}
\end{figure}

For a realisation, $\omega$, of $(N^{e,*})_{e \in \mathcal E(G),* \in \{\cross,\dbars\}}$ and a finite subgraph, $\Gfin$, of $G$ we construct loops in the usual way (see e.g. \cite{U}). In fact, the construction only depends on the events of $\omega$ belonging to $e \in \mathcal E(\Gfin)$. 
More precisely, a loop is the support of a parametrisation $[0,1] \to \mathcal V(\Gfin)\times[0,\beta)_\textnormal{per}$ that respects the links $\omega_i=(e_i,t_i,*_i)$ of $\omega=(\omega_i)_i$ with $e_i \in \mathcal E(\Gfin)$. This means that we start at some $(x,t_0)$ and move along $x$ in some time direction until we encounter a link on an edge $e_i=\{x,y\}\in \mathcal E(\Gfin)$ at time $t_i$. Then we jump to $y$ and continue in the same time direction, if $*_i=\cross$, or in the opposite time direction, if $*_i=\dbars$, as before. If we reach $(x,0)$ or $(x,\beta)$ we use periodicity to move to $(x,\beta)$ or $(x,0)$, respectively, and continue in the same direction. This procedure is best explained by a picture (see figure \ref{fig:loopexample}). From this procedure we obtain a partition, $\mathcal L_{\Gfin}(\omega)$, of $\mathcal V(\Gfin)\times[0,\beta)_\textnormal{per}$ into loops.
By $L_{\Gfin}(\omega):=|\mathcal L_{\Gfin}(\omega)|$ we denote the number of loops within $\Gfin$ for the realisation $\omega$. Finally, for $\theta \geq 1$, the probability measure of interest is given by
\begin{equation}
\mathbb P_{\Gfin}^\theta(B) := \frac{\mathbb E_{\Gfin}\left[\mathbf{1}_{B} \theta^{L_{\Gfin}} \right] }{\mathbb E_{\Gfin}\left[ \theta^{L_{\Gfin}} \right]}.
\label{eq:P-theta-definition}
\end{equation}

Our first theorem concerns the $d$-regular tree. We denote by $T_x^{n}$ the $d$-regular tree rooted at $x$ and containing ${n}$ generations. 
For a given finite tree $\Tfin$ and $x \in \mathcal V(\Tfin)$ we consider the event $E^{x\to m}_{\Tfin}$ that there is a loop within $\mathcal L_{\Tfin}(\omega)$ containing $x$ (at some time) that reaches the $m^\text{th}$ generation of the tree.

\begin{thrm}\label{thrm:ex-infinite-loops}
	Let $b >a>\theta>q$ be arbitrary but fixed. Then:
	\begin{enumerate}
		\item There is a $d_0 \in \mathbb N$, depending on $a$ and $b$, such that for all $d\geq d_0$ and all $\beta=\beta(d)\in \left[ \frac{a}{d},\frac{b}{d}\right]$ we have
	\begin{align*}
	\liminf_{m \to \infty} \inf_{n \geq m} \mathbb P_{T_{r}^{n}}^\theta [E^{r\to m}_{T_r^n}]>0   .
	\end{align*}
	\item There is a $d_0 \in \mathbb N$, depending on $q$, such that for all $d\geq d_0$ and all
	$\beta=\beta(d)\leq \frac{q}{d}$ we have
	\begin{align*}
	\limsup_{m \to \infty} \sup_{n \geq m} \mathbb P_{T_r^{n}}^\theta [E^{r\to m}_{T_r^n}]=0  . 
	\end{align*}
	\end{enumerate}
\end{thrm}
Hence, under the conditions of part 1, with positive probability there is an infinite loop from the root in the limit $n\to\infty$.
Note that the existence of the limit in the theorem is not proved in general (although they certainly exist) and only limited results exist, such as for the case $u=1/2$ and $\theta=2$ \cite{B-L-U}. For this reason we take the $\liminf$ and $\limsup$, respectively.

Our next result concerns the \GW tree. Let $X$ be a random variable with values in $\mathbb{N}_0$. Denote by $\PGW$ and $\EGW$ the probability and expectation with respect to the \GW tree with offspring distribution $X$, respectively. We denote by $T^{X}_{r,n}$ a realisation of the \GW tree with root $r$ cut at level $n$ (i.e., the tree has $n$ generations). We consider the \emph{quenched} measure
\begin{equation}
\Pqu:=\EGW\left[\mathbb P_{T^{X}_{r,n}}^\theta(\cdot)\right].
\end{equation}
We have the following theorem:
\begin{thrm} $~$\label{thrm:GWloops}
\vspace{-6pt}
\begin{enumerate}
		\item If there is an $\varepsilon>0$ and a $\beta>0$, both depending on the distribution of $X$, such that
\begin{align*}
\EGW &\left[ \mathrm{e}^{-\frac{\beta}{\theta}X}\right]\leq 1-\varepsilon
\intertext{and}
\EGW &\left[ \left(\frac{\theta^2+\beta \theta}{\theta^2+\mathrm{e}^{\beta\theta}-1}\right)^{X}-\left(\mathrm{e}^{-\frac{\beta}{\theta}}\left(1+\frac{\beta}{\theta}(1-\varepsilon)\right)\right)^{X}\right]\geq \varepsilon.
\end{align*}
Then
\begin{align*}
\liminf_{m \to \infty} \inf_{n \geq m} \Pqu [E^{r\to m}_{T^X_{r,n}}]>0.
\end{align*}
\item If there is a $\beta>0$, depending on the distribution of X, such that
\begin{align*}
\EGW\left[X \mathrm{e}^{-\frac{\beta}{\theta} X} \left(1+\frac{\mathrm{e}^{\beta \theta}-1}{\theta^2}\right)^{X-1}\right] \frac{\mathrm{e}^{\beta \theta}-1}{\theta^2}< 1
\end{align*}
holds, then
\begin{align*}
\limsup_{m \to \infty} \sup_{n \geq m} \Pqu [E^{r\to m}_{T^X_{r,n}}]=0.
\end{align*}
\end{enumerate}
\end{thrm}

To illustrate how to make use of the conditions within Theorem \ref{thrm:GWloops} we have the following example.
\begin{ex}\label{ex:poisson-offspring}
	Consider Poisson distributed offspring $X \sim \operatorname{Poi}(\mu)$ with $\mu>0$.
	\begin{enumerate}
		\item Let us fix $a>\theta$ and choose $\varepsilon\leq \frac{1}{2}$ such that 
		$1-\exp\left(-\frac{a}{\theta}\varepsilon\right)>\varepsilon$. Then for $\beta:=\frac{a}{\mu}$ it holds
		\begin{align}
		&\EGW\left[ \left(\frac{\theta^2+\beta\theta}{\theta^2+\mathrm{e}^{\beta\theta}-1}\right)^X
		-\left(\mathrm{e}^{-\frac{\beta}{\theta}}\left(1+\frac{\beta}{\theta}(1-\varepsilon)\right)\right)^X\right]
		\\
		&=\mathrm{e}^{-\mu\left(1-\frac{\theta^2+\beta\theta}{\theta^2+\mathrm{e}^{\beta\theta}-1}\right)}-
		\mathrm{e}^{-\mu\big(1-\mathrm{e}^{-\frac{\beta}{\theta}}\big(1+\frac{\beta}{\theta}(1-\varepsilon)\big)\big)}
		\\
		&
		\underset{\mathclap{\beta=a/\mu}}{\overset{\mathclap{\mu \to \infty}}{\longrightarrow}} \; 1-\exp\left(-\frac{a}{\theta}\varepsilon\right).
		\end{align}
		Now by choosing $\mu \geq \frac{a}{\theta}$ large enough and estimating
		\begin{align}
		&\EGW\left[\mathrm{e}^{-\frac{\beta}{\theta}X}\right]=\exp\left(-\mu\left(1-\mathrm{e}^{-\frac{\beta}{\theta}}\right)\right)
		\\
		&\leq \exp\left(-\frac{\mu\frac{\beta}{\theta}}{1+\frac{\beta}{\theta}}\right)
		\overset{\beta \leq \theta,}{\underset{\frac{1}{2}\geq\varepsilon}{\leq}} 
		\exp\left(-\frac{a}{\theta}\varepsilon\right)<1-\varepsilon
		\end{align}
		we see that -- with positive probability -- there are long loops.
		\item On the other hand, by calculating
		\begin{align}
		&\EGW\left[X \mathrm{e}^{-\frac{\beta}{\theta} X}\left(1+\frac{\mathrm{e}^{\beta\theta}-1}{\theta^2}\right)^{X-1} \right] 
		\frac{\mathrm{e}^{\beta\theta}-1}{\theta^2}
		\\
		&= \mu \mathrm{e}^{-\frac{\beta}{\theta}}\frac{\mathrm{e}^{\beta\theta}-1}{\theta^2}
		\exp\left[\mu\left(\mathrm{e}^{-\frac{\beta}{\theta}}
		\left(1+\frac{\mathrm{e}^{\beta\theta}-1}{\theta^2}\right)-1\right)\right] 
		\overset{\beta \to 0}{\longrightarrow} 0
		\end{align}
		for fixed $\mu$ we see that there is a $\beta_0=\beta_0(\mu,\theta)$ such that for all $\beta \leq \beta_0$ there almost surely are no infinite loops.
	\end{enumerate}
\end{ex}

The following corollary gives further sufficient conditions for the distribution of $X$ and ranges of $\beta$ where Theorem \ref{thrm:GWloops} is applicable. In particular, we will be able to deduce Theorem \ref{thrm:ex-infinite-loops} from this corollary by considering a deterministic offspring distribution.

\newpage
\begin{cor} $~$\label{cor:GWloops-large-expectation}
	\vspace{-6pt}
	\begin{enumerate}
		\item 
		Fix $b\geq a > \theta$ and let $X$ be an integrable offspring distribution with $\PGW[X>0]>0$. 
		Furthermore, choose $c_2 \geq c_1>0$ such that $B_X:=\left[c_1\leq \frac{X}{\EGW[X]}\leq c_2\right]$ has positive probability.
		Then there is a $\lambda_0\in \mathbb N$ such that for all $\lambda \in \mathbb N$ with $\lambda \geq \lambda_0$ on 
		the \GW tree with rescaled offspring distribution $\lambda \cdot X$ and for every 
		$\beta \in \frac{1}{\EGW[\lambda X]} \frac{1}{c_1 \PGW(B_X)} [a,b]$ we have
		\begin{align*}
		\liminf_{m \to \infty} \inf_{n \geq m} \mathbf{P}^{\theta,\lambda\cdot X}_{n} [E^{r\to m}_{T^{\lambda \cdot X}_{r,n}}]>0.
		\end{align*}
		\item Fix $q<\theta$. Then there is a $d_0 \in \mathbb N$ such that for every $d \geq d_0$, every offspring distribution $X$ bounded above by $d$ and for all $\beta \leq \frac{q}{d}$ it holds
		\begin{align*}
		\limsup_{m \to \infty} \sup_{n \geq m} \Pqu [E^{r\to m}_{T^X_{r,n}}]=0.
		\end{align*}
	\end{enumerate}
\end{cor}

Note that within the first part of this corollary larger choices of $c_2$ will make $\lambda_0$ larger. 
Moreover, for given $X$ and $c_2$ one should seek to maximize the quantity $c_1 \PGW(B_X)\leq 1$ to show existence 
of long loops for $\beta$ just above $\frac{\theta}{\EGW[\lambda X]}$.

\section{Proofs}

For the next three results we will consider the following setting and notation:
\begin{align}\label{eq:tree-subtrees-setting}
\begin{split}
&\text{Let $T_0$ be a fixed finite tree rooted in $r$ and denote the children}\\
&\text{of $r$ by $x_1,\ldots,x_d$. Furthermore, write $T_j$ for the subtree of $T_0$}\\
&\text{rooted in $x_j$, $j=1,\ldots,d$.}
\end{split}
\end{align}

We begin by giving estimates on quantities of $T_0$ in terms of the corresponding quantities on its subtrees $T_j$. In particular, Proposition \ref{prop:loopnumbers-subgraphs} deals with the number $L_{T_0}$ of loops and Corollary \ref{cor:part-func-estimation} with estimating the partition function of our model.

This will enable us to prove occurrence of long loops with the main ingredients being the recursive estimation in Lemma \ref{lem:P-of-A-cap-B-recursion} and a certain self-similarity argument within the proof of Lemma \ref{lem:GWzeta-recursion}.

Similarly, to show absence of long loops we will prove exponential decay of the probability for a loop to reach generation $m$ by a recursive argument and using the same kind of self-similarity as above in Lemma \ref{lem:GWexpdecay}.

The following proposition is a basic observation and we will make use of it multiple times.
\begin{prop}\label{prop:loopnumbers-subgraphs}
	Given (\ref{eq:tree-subtrees-setting}), we have
	\begin{align}
	-\sum_{j=1}^d N_\beta^{\{r,x_j\}} \leq L_{T_0}- \bigg(\sum_{j=1}^d L_{T_j} + 1 \bigg) \leq \sum_{j=1}^d \big(|N_\beta^{\{r,x_j\}}-1|-1\big). \label{eq:loopnumbers-subgraphs}
	\end{align}
\end{prop}

Note that by $L_{T_j}(\omega)$ we mean the number of loops on $T_j$ in a realisation $\omega$, i.e.\ only links on the edges $e \in \mathcal E(T_j)$ are taken into account for constructing these loops. \\
Furthermore, (\ref{eq:loopnumbers-subgraphs}) becomes an equality for $N_{\beta}^{\{r,x_j\}}\in \{0,1\}$.
\begin{proof}
The result is immediate once we understood how adding a link to a realisation changes the number of loops. If there are no links between $r$ and its children the number of loops is given by
\begin{align}
	L_{T_0} = \sum_{j=1}^d L_{T_j} + 1. 
	\end{align}
Adding 
the first link to an edge will merge the loops at its end points, reducing the number of loops by $1$. Hence, we have $L_{T_0}=\sum_{j=1}^d L_{T_j}+1-\sum_{j=1}^d {N_\beta^{\{r,x_j\}}}$ in the case that $N_\beta^{\{r,x_j\}}\in\{0,1\}$ holds for all $j$.\\
Any further link may either merge two loops, split a loop into two or alter a single loop. Therefore the number of loops changes by at most one and the result follows. 
\end{proof}
The previous estimations on the loop numbers enable us to give estimates on the partition function:
\begin{cor}\label{cor:part-func-estimation}
	In the setting of (\ref{eq:tree-subtrees-setting})
	we obtain the following bounds
\begin{flalign}
&& \mathbb E_{T_0}[\theta^{L_{T_0}}] &\geq 
\theta \mathrm{e}^{-\beta d+\beta d/\theta} \prod_{j=1}^d \mathbb E_{T_j}\left[\theta^{L_{T_j}}\right] & \label{eq:part-func-est-below}
\\&\text{and}& \mathbb E_{T_0}[\theta^{L_{T_0}}]
&\leq \theta \mathrm{e}^{-\beta d}\left(1+\frac{\mathrm{e}^{\beta \theta}-1}{\theta^2}\right)^d\prod_{j=1}^d \mathbb E_{T_j}[\theta^{L_{T_j}}].
\label{eq:part-func-est-above}
\end{flalign}
\end{cor}
Note that we assumed $\theta \geq 1$ in the definition of our model. One could define $\mathbb P^\theta_{\Gfin}$ in the same way for $0<\theta \leq 1$ and up to now this would only exchange the upper and lower bound for the partition function in (\ref{eq:part-func-est-below}) and (\ref{eq:part-func-est-above}).
\begin{proof}
	We calculate using the lower bound from Proposition \ref{prop:loopnumbers-subgraphs}:
\begin{align}
	&\mathbb E_{T_0}[\theta^{L_{T_0}}] 
	= \sum_{n_1\in \mathbb N_0} \cdots \sum_{n_d\in \mathbb N_0} \mathbb E_{T_0} \left [  \mathbf{1}{\left[N_\beta^{\{r,x_j\}}=n_j \, \forall j\right]} \cdot \theta^{L_{T_0}} \right ] 
	\\
	&\geq \sum_{n_1\in \mathbb N_0} \cdots \sum_{n_d\in \mathbb N_0} \mathbb E_{T_0} \left [  \prod_{j=1}^d \left(  \mathbf{1}{\left[N_\beta^{\{r,x_j\}}=n_j\right]}\right) \prod_{j=1}^d \left(\theta^{L_{T_j}}\right) \right ] \theta^{-\sum_{j=1}^d n_j +1 }.
\end{align}
Now the factors in the expectation are independent, therefore we obtain
\begin{align}
	\mathbb E_{T_0}[\theta^{L_{T_0}}]\geq& \theta \cdot \prod_{j=1}^d \left(\sum_{n_j\in \mathbb N_0} \rho_{T_0}\left[N_\beta^{\{r,x_j\}}=n_j\right] \mathbb E_{T_0} \left[ \theta^{L_{T_j}} \right] \theta^{-n_j} \right)\\
	=& \theta \cdot \left( \sum_{k \in \mathbb N_0} \frac{\beta^k}{k!}\mathrm{e}^{-\beta} \theta^{-k} \right)^d \prod_{j=1}^d \mathbb E_{T_0}\left[\theta^{L_{T_j}}\right]\\
	=& \theta \cdot \mathrm{e}^{-\beta d+\beta d/\theta} \prod_{j=1}^d \mathbb E_{T_j}\left[\theta^{L_{T_j}}\right].
\end{align}
Here we used that by the construction of the model we have $\mathbb E_{T_0}[\theta^{L_{T_j}}]=\mathbb E_{T_j}[\theta^{L_{T_j}}]$ as $L_{T_j}$ only depends on the links on $\mathcal E(T_j)$.\\
The upper bound on the partition function follows similarly from the second inequality in (\ref{eq:loopnumbers-subgraphs}):
\begin{align}
\mathbb E_{T_0}[\theta^{L_{T_0}}] &\leq \sum_{n_1\in \mathbb N_0} \cdots \sum_{n_d\in \mathbb N_0} \prod_{j=1}^d \left( \rho_{T_0}\left[N_\beta^{\{r,x_j\}}=n_j\right] \mathbb E_{T_0}\left[\theta^{L_{T_j}}\right] \theta^{|n_j-1|-1} \right) \theta 
\\
&=\theta \mathrm{e}^{-\beta d}\left(1+\sum_{k=1}^\infty \frac{\beta^k}{k!}\theta^{k-2}\right)^d \prod_{j=1}^d \mathbb E_{T_j}[\theta^{L_{T_j}}]
\\
&=\theta \mathrm{e}^{-\beta d}\left(1+\frac{\mathrm{e}^{\beta \theta}-1}{\theta^2}\right)^d\prod_{j=1}^d \mathbb E_{T_j}[\theta^{L_{T_j}}].
\end{align}
\end{proof}

The next corollary gives estimates on three events that are of particular interest.
\begin{cor}\label{cor:PofA-estimates}
	Given (\ref{eq:tree-subtrees-setting}) we define $A_\emptyset:=[N_\beta^{\{r,x_j\}}=0 \; \forall j\leq d]$
	and
	\begin{align} 
		A &:=[N_\beta^{\{r,x_j\}}\leq 1 \; \forall j\leq d]. \label{eq:def-A}
	\end{align}
	Then
	\begin{align}
	\mathbb P_{T_0}^\theta(A_\emptyset) \leq \mathrm{e}^{-\beta d/\theta} \qquad \text{and} \qquad \mathbb P_{T_0}^\theta(A) \geq \left(\frac{\theta^2+\beta \theta}{\theta^2+\mathrm{e}^{\beta\theta}-1}\right)^d.
	\end{align}
	\begin{proof}
		Using (\ref{eq:loopnumbers-subgraphs}) and (\ref{eq:part-func-est-below}) we can estimate
		\begin{align}
		\mathbb P_{T_0}^\theta(A_\emptyset) &=\frac{\mathbb E_{T_0}\left[\mathbf{1}_{A_\emptyset} \theta^{L_{T_0}}\right]}{\mathbb E_{T_0}\left[\theta^{L_{T_0}}\right]} 
		\leq \frac{\theta (\mathrm{e}^{-\beta})^d \prod_{j=1}^d \mathbb E_{T_j}[\theta^{L_{T_j}}]}{\theta \mathrm{e}^{-\beta d+\beta d/\theta}\prod_{j=1}^d \mathbb E_{T_j}[\theta^{L_{T_j}}]} = \mathrm{e}^{-\beta d/\theta}.
		\end{align}
		Similarly, from (\ref{eq:loopnumbers-subgraphs}) and (\ref{eq:part-func-est-above})
		we have
		\begin{align}
		\mathbb P_{T_0}^\theta(A) 
		&\geq \frac{\theta\left(\mathrm{e}^{-\beta} + \mathrm{e}^{-\beta} \beta \theta^{-1}\right)^d \prod_{j=1}^d \mathbb E_{T_j}\left[\theta^{L_{T_j}}\right] }{\theta \mathrm{e}^{-\beta d}\left(1+\frac{\mathrm{e}^{\beta\theta}-1}{\theta^2}\right)^d \prod_{j=1}^d \mathbb E_{T_j}[\theta^{L_{T_j}}]}
		\\
		&=\left(\frac{\theta^2+\beta \theta}{\theta^2+\mathrm{e}^{\beta\theta}-1}\right)^d.
		\end{align}		
	\end{proof}
\end{cor}

\subsection{Occurrence of long loops}\label{sec:longloops}
To prove existence of long loops we construct a recursive estimation for the probability to reach level $m$. The following lemma provides a key relation to obtain this.
\begin{lem}\label{lem:P-of-A-cap-B-recursion}
In the setting of (\ref{eq:tree-subtrees-setting}) let us consider the events
	\begin{align}
	\begin{split}
	B_{T_j}^{x \not \to m}:=\big \{\omega : &\text{ There is a loop in }\mathcal L_{T_j}(\omega)\text{ containing $(x,s)$ for }\\
	&\text{some $s\in[0,\beta)_\textnormal{per}$ that fails to reach generation }m \big \}\label{eq:B-def}
	\end{split}
	\end{align}
	with $x \in \mathcal V(T_j)$, $m\in \mathbb N$ and $0\leq j \leq d$. Then
	\begin{align}
	\mathbb P^\theta_{T_0}(A \cap B_{T_0}^{r \not \to m}) \leq \; \sum_{\mathclap{J \subseteq \{1,\ldots,d\}}} \; \mathrm{e}^{-\beta d/\theta} \left(\frac{\beta}{\theta}\right)^{|J|} \prod_{j\in J} \mathbb P^\theta_{T_j}(B_{T_j}^{x_j \not \to m-1})
	\end{align}
	holds, where $A$ is defined by (\ref{eq:def-A}).
\end{lem}
\begin{proof}
	For $(x,t)\in \mathcal V(T_0)\times [0,\beta)_\textnormal{per}$, $m \in \mathbb N$ and $0 \leq j \leq d$ let us define the event
	\begin{align}
	\begin{split}
		B_{T_j}^{(x,t)\not\to m}:=\big\{ \omega : &\text{ The loop within }\mathcal L_{T_j}(\omega)\text{ that contains }(x,t)
		\\
		&\qquad\qquad\qquad\quad\text{ fails to reach generation }m\big\},
	\end{split}
	\end{align}
	therefore we have and $B_{T_j}^{(x,t)\not \to m} \subseteq B_{T_j}^{x\not \to m}$ for every time $t$. For any subset $J \subseteq \{1,\ldots,d\}$ we may also write 
	\begin{equation}
	A_{J}:= \left [N_\beta^{\{r,x_j\}}=\left \lbrace \begin{array}{cl} 1 & \text{ if }j \in J \\ 0 & \text{ else} \end{array} \right. \right ] .\label{eq:Ak-def}
	\end{equation}
	Now let us fix $\omega \in A_{J}$ and for $j \in J$ denote by $t_j=t_j(\omega)$ the time of the link on $\{r,x_j\}$. Since the loop $\gamma_{(r,0)} \in \mathcal L_{T_0}(\omega)$ containing $(r,0)$ fails to reach generation $m$ iff for all $j \in J$ the loop within $\mathcal L_{T_j}(\omega)$ containing $(x_j,t_j(\omega))$ fails to reach down $m-1$ generations, we have
	\begin{align}
	\mathbf{1}\left[{B_{T_0}^{r\not \to m}}\right](\omega)&=\mathbf{1}\left[{B_{T_0}^{(r,0)\not\to m}}\right](\omega) \label{eq:AJ-single-looproot}\\
	&= \prod_{j=1}^k \mathbf{1}\left[{B_{T_j}^{(x_j,t_j(\omega))\not\to m-1}}\right](\omega) \leq \prod_{j=1}^k \mathbf{1}\left[{B_{T_j}^{x_j \not \to m-1}}\right](\omega).
	\end{align}
	Here, (\ref{eq:AJ-single-looproot}) holds as, by $\omega \in A_J$, edges containing $r$ have at most one link, hence $\{r\}\times [0,\beta)$ is contained in one single loop.
	Therefore, using Proposition \ref{prop:loopnumbers-subgraphs}, we obtain
	\begin{equation}
	\begin{split}
	&\mathbb E_{T_0}\left[\mathbf{1}_{A_{J}} \mathbf{1}\left[{B_{T_0}^{r \not \to m}}\right] \, \theta^{L_{T_0}}\right] \\
	&\leq\mathbb E_{T_0} \left [ \mathbf{1}_{A_J} \prod_{j\in J} \left(\mathbf{1}\left[{B_{T_j}^{x_j \not \to m-1}}\right] \,\theta^{L_{T_j}} \right) \prod_{j\in J^\com} \left(\theta^{L_{T_j}} \right) \right ] \cdot \theta^{-|J|+1}.
	\end{split}
	\end{equation}
	By independence, this yields
	\begin{align}
	\begin{split}
	&\mathbb E_{T_0}\left[\mathbf{1}_{A_J} \mathbf{1}\left[{B_{T_0}^{r\not\to m}}\right] \, \theta^{L_{T_0}}\right]
	\\
	&\leq\theta^{1-|J|} \cdot \rho_{T_0}(A_J) \cdot \prod_{j\in J} \mathbb E_{T_0} \bigg [ \mathbf{1}\left[{B_{T_j}^{x_j \not \to m-1}}\right] \cdot \theta^{L_{T_j}} \bigg ] \cdot \prod_{j\in J^\com} \mathbb E_{T_0} \bigg [ \theta^{L_{T_j}} \bigg ]
	\end{split}
	\\
	&=\theta^{1-|J|} \cdot \beta^{|J|} \mathrm{e}^{-\beta d}\cdot \prod_{j\in J}  \mathbb{P}_{T_j}^\theta \big( B_{T_j}^{x_j \not \to m-1}\big) \mathbb{E}_{T_j} \bigg [\theta^{L_{T_j}} \bigg ] \cdot \prod_{j\in J^\com}  \mathbb{E}_{T_j} \bigg [ \theta^{L_{T_j}} \bigg ].
	\end{align}
	Using Corollary \ref{cor:part-func-estimation} we now see that
	\begin{align}
	\begin{split}
	&
	\mathbb P_{T_0}^\theta\left(A_J \cap B_{T_0}^{r\not\to m}\right)
	\\
	&=  \frac{\mathbb E_{T_0}\left[\mathbf{1}_{A_J} \mathbf{1}\left[{B_{T_0}^{r\not\to m}}\right] \, \theta^{L_{T_0}} \right]}{\mathbb E_{T_0} \left[ \theta^{L_{T_0}} \right] } 
	\end{split}
	\\
	&\leq \frac{\theta^{1-|J|} \cdot \beta^{|J|} \mathrm{e}^{-\beta d}\cdot \prod_{j\in J}  \mathbb{P}_{T_j}^\theta \big( B_{T_j}^{x_j \not \to m-1}\big) \cdot \prod_{j=1}^d \mathbb E_{T_j} \bigg [ \theta^{L_{T_j}} \bigg ]}{\theta \cdot \mathrm{e}^{-\beta d+\beta d/\theta} \prod_{j=1}^d \mathbb E_{T_j}\left[\theta^{L_{T_j}}\right]} 
	\\
	&= \mathrm{e}^{-\beta d/\theta} \left( \frac{\beta}{\theta} \right)^{|J|} \prod_{j\in J}  \mathbb{P}_{T_j}^\theta \big( B_{T_j}^{x_j \not \to m-1}\big) .
	\end{align}
		Therefore, we are able to conclude
		\begin{align} 
		\mathbb P_{T_0}^\theta \left( A \cap B_{T_0}^{r\not\to m} \right) &=\; \sum_{\mathclap{J \subseteq \{1,\ldots,d\}}} \; \mathbb P_{T_0}^\theta\left(A_J \cap B_{T_0}^{r\not\to m}\right)  \label{eq:ABcest}
		\\
		&\leq \; \sum_{\mathclap{J \subseteq \{1,\ldots,d\}}} \; \mathrm{e}^{-\beta d/\theta} \left( \frac{\beta}{\theta} \right)^{|J|} \prod_{j\in J}  \mathbb{P}_{T_j}^\theta \big( B_{T_j}^{x_j \not \to m-1}\big).
		\end{align}
\end{proof}

We now turn our attention to the \GW tree and show the crucial recurrence relation.

\begin{lem} \label{lem:GWzeta-recursion}
	Consider the \GW tree with offspring distribution $X$. Define
	\begin{align}
	\zeta^m_n:=\Pqu[B_{T^X_{r,n}}^{x \not \to m}].
	\end{align}
	The following recursion holds:
	\begin{align}
	1-\zeta^m_n
	\geq \EGW \left[ \left(\frac{\theta^2+\beta \theta}{\theta^2+\mathrm{e}^{\beta\theta}-1}\right)^X - \left(\mathrm{e}^{-\beta/\theta} \left(1+ \frac{\beta}{\theta} 
	\zeta^{m-1}_{n-1}\right) \right)^X \right].
	\end{align}
\end{lem}
\begin{proof}
	To begin denote by $X_r$ the number of offspring of the root and fix $n,m \in \mathbb N$ with $n \geq m$. For a realisation of the \GW tree, define $T_0:=T^X_{r,n}$ to be the (random) tree rooted in $r$ that is cut at level $n$. Then we observe that
	\begin{equation}
	1-\zeta^m_n 
	=\sum_{d \in \mathbb N_0}\EGW\left[\mathbb P_{T_0}^\theta \left((B_{T_0}^{x \not \to m})^\com\right)\bigg| X_r=d\right]\PGW[X_r=d].
	\end{equation}
	By using the notation from (\ref{eq:tree-subtrees-setting}) for $X_r=d$ and applying Corollary \ref{cor:PofA-estimates}, we obtain
	\begin{align}
	\mathbb P_{T_0}^\theta \left((B_{T_0}^{x \not \to m})^\com\right)
	&\geq \mathbb P_{T_0}^\theta \left( A \cap (B_{T_0}^{x \not \to m})^\com \right)
	\\
	&= \mathbb P_{T_0}^\theta(A) - \mathbb P_{T_0}^\theta\big(A \cap B_{T_0}^{r\not\to m}\big) 
	\\
	&\geq
	\left(\frac{\theta^2+\beta \theta}{\theta^2+\mathrm{e}^{\beta\theta}-1}\right)^d - \mathbb P_{T_0}^\theta\big(A \cap B_{T_0}^{r\not\to m}\big).
	\end{align}
	The second term can be estimated by using Lemma \ref{lem:P-of-A-cap-B-recursion}:
	\begin{align}
	\mathbb P_{T_0}^\theta\big(A \cap B_{T_0}^{r\not\to m}\big) \leq \; \sum_{\mathclap{J \subseteq \{1,\ldots,d\}}} \; \mathrm{e}^{-\beta d/\theta} \left(\frac{\beta}{\theta}\right)^{|J|} \prod_{j\in J} \mathbb P^\theta_{T_j}\left(B_{T_j}^{x_j \not \to m-1}\right).
	\end{align}
	Hence by taking the expectation we have 
	\begin{align}
	&\EGW\left[\mathbb P_{T_0}^\theta \left(A \cap B_{T_0}^{x \not \to m}\right) \bigg| X_r=d\right]
	\\
	&\leq \sum_{\mathclap{J \subseteq \{1,\ldots,d\}}} \; \mathrm{e}^{-\beta d/\theta} \left(\frac{\beta}{\theta}\right)^{|J|} \prod_{j\in J} \EGW\left[\mathbb{P}_{T_j}^\theta \left( B_{T_j}^{x_j \not \to m-1}\right)\right].
	\intertext{By self-similarity in expectation, this yields}
	&= \sum_{k=0}^d {d\choose k} \mathrm{e}^{-\beta d/\theta}\left(\frac{\beta}{\theta}\right)^k \left(\zeta^{m-1}_{{n-1}}\right)^k
	\\
	&=\mathrm{e}^{-\beta d/\theta} \left(1+ \frac{\beta}{\theta}\zeta_{{n}-1}^{m-1}\right)^d.
	\end{align}
	The result follows.
\end{proof}

We can now use this recursion to prove the occurrence of long loops in the \GW tree under our assumptions on the distribution of $X$.
\begin{proof}[Proof of Theorem \ref{thrm:GWloops}, part 1.]
	We proceed via induction on $m$ to prove that, for all $m \in \mathbb N$ and all ${n} \geq m$, we have 
	\begin{equation}
	1-\zeta_{n}^m \geq \varepsilon.
	\end{equation}
	This is sufficient as we have $(B_{T^X_{r,n}}^{x\not \to m})^\com\subseteq E^{r\to m}_{T^X_{r,n}}$.
	For the base step note that, given $X_r=d$, we have $B_{T^X_{r,n}}^{r \not \to 1}=A_\emptyset$. Therefore we have
	\begin{align}
	\zeta_{n}^{1}&=\sum_{d \in \mathbb N_0} \EGW\left[ \mathbb P^\theta_{T^X_{r,n}}(A_\emptyset) \bigg |X_r=d \right] \PGW[X_r=d] \\
	&\leq \sum_{d \in \mathbb N_0}\mathrm{e}^{-\frac{\beta d}{\theta}}\PGW[X_r=d]\leq 1-\varepsilon,
	\end{align}
	where the first inequality uses Corollary \ref{cor:PofA-estimates} and the last inequality follows from our assumption on the distribution of $X$. Assume now the induction hypothesis holds for $m-1$ and all $\tilde {n} \geq m-1$. By Lemma \ref{lem:GWzeta-recursion}, our induction hypothesis and our assumptions on the distribution of $X$ we can estimate
	\begin{align}
	1-\zeta_{n}^m &\geq \EGW \left[\left(\frac{\theta^2+\beta \theta}{\theta^2+\mathrm{e}^{\beta\theta}-1}\right)^X -  \left(\mathrm{e}^{-\beta /\theta} \left(1+ \frac{\beta}{\theta}\zeta_{n-1}^{m-1}\right)\right)^X \right]
	\\
	&\geq \EGW \left[\left(\frac{\theta^2+\beta \theta}{\theta^2+\mathrm{e}^{\beta\theta}-1}\right)^X -  \left(\mathrm{e}^{-\beta /\theta} \left(1+ \frac{\beta}{\theta}(1-\varepsilon)\right)\right)^X \right]
	\\
	&\geq \varepsilon.
	\end{align}
\end{proof}

\begin{proof}[Proof of Corollary \ref{cor:GWloops-large-expectation}, part 1]
	For $\lambda\in \mathbb N$ write $X_\lambda:=\lambda\cdot X$, $\bar d_\lambda:=\EGW[X_\lambda]$. Now choose $\varepsilon>0$ such that the inequalities
	\begin{flalign}
	&& \EGW\left[1-\exp\left(-\frac{a}{\theta c_1 \PGW(B_X)}\frac{X}{\EGW[X]} \right)\right]&\geq \varepsilon& 
	\label{eq:Xc-ind-start}
	\\
	&\text{and}& 
	\exp\left(-\frac{{a}}{\theta}\frac{\varepsilon}{\PGW(B_X)}\right) &< 1-\frac{\varepsilon}{\PGW(B_X)}
	\end{flalign}
	are fulfilled. Using $\PGW[X>0]>0$ and $\frac{{a}}{\theta}>1$, this possible. 
	Since
	\begin{align}
	1\geq\left(1+\frac{1}{\theta^2}\sum_{k=2}^\infty \frac{ \bar d_\lambda^{-k}\left(\frac{b \theta}{c_1 \PGW(B_X)}\right)^k}{k!} \right)^{-\bar d_\lambda} \longrightarrow 1 
	\end{align}
	holds as $\bar d_\lambda=\lambda\cdot \EGW[X]\overset{\lambda \to \infty}{\longrightarrow} \infty$, we may find a $\lambda_0 \in \mathbb N$ such that for all $\lambda \geq \lambda_0$ we have
	\begin{align}
	&\left(1+\frac{1}{\theta^2}\sum_{k=2}^\infty \frac{\left(\frac{b\theta}{c_1 \PGW(B_X) \bar d_\lambda}\right)^k}{k!} \right)^{-\bar d_\lambda \cdot c_2 } 
	\\
	&\geq 1-
	\underbrace{\left[
		1-\frac{\varepsilon}{\PGW(B_X)}-\exp\left(-\frac{a}{\theta}\frac{\varepsilon}{\PGW(B_X)}\right) \right]}_{>0 \text{ by choice of }\varepsilon} 
	\\
	&=\frac{\varepsilon}{\PGW(B_X)}+\exp\left(-\frac{a}{\theta}\frac{\varepsilon}{\PGW(B_X)}\right). \label{eq:cor4-lambda-estimate}
	\end{align}
	Then, by the lower bound on $\beta$ and (\ref{eq:Xc-ind-start}), we have
	\begin{align}
	\EGW\left[\mathrm{e}^{-\frac{\beta}{\theta}X_\lambda}\right] \leq \EGW\left[\exp\left(-\frac{1}{\theta}\frac{a}{c_1 \PGW(B_X) \lambda \EGW[X]} \lambda X \right)\right] \leq 1-\varepsilon,
	\end{align}
	therefore $X_\lambda$ satisfies the first condition of Theorem \ref{thrm:GWloops}, part 1.\\
	Furthermore, on $B_X=[c_1 \bar d_\lambda \leq X_\lambda \leq c_2 \bar d_\lambda]$ and taking $\beta \in \frac{1}{\bar d_\lambda} \frac{1}{c_1 \PGW(B_X)} [a,b]$, we can estimate
	\begin{align}
	\left(\frac{\theta^2+\beta \theta}{\theta^2+\mathrm{e}^{\beta\theta}-1}\right)^{X_\lambda} 
	&= \left( 1 + \frac{1}{\theta^2+\beta \theta} \sum_{k=2}^\infty \frac{(\beta\theta)^k}{k!} \right)^{-X_\lambda} 
	\\
	&\geq \Bigg(1+\frac{1}{\theta^2}\sum_{k=2}^\infty \frac{\left(\frac{{b}\theta}{\bar d_\lambda c_1 \PGW(B_X) }\right)^k}{k!}\Bigg)^{-c_2 \bar d_\lambda}
	\\
	&\overset{\mathclap{(\ref{eq:cor4-lambda-estimate})}}{\geq} \frac{\varepsilon}{\PGW(B_X)}+\exp\left(-\frac{{a}}{\theta}\frac{\varepsilon}{\PGW(B_X)}\right)
	\intertext{as well as}
	\left(\mathrm{e}^{-\frac{\beta}{\theta}}\left(1+\frac{\beta}{\theta}(1-\varepsilon)\right) \right)^{X_\lambda} &=\bigg( \mathrm{e}^{-\frac{\beta}{\theta}\varepsilon} \underbrace{\mathrm{e}^{-\frac{\beta}{\theta}(1-\varepsilon)}\left(1+\frac{\beta}{\theta}(1-\varepsilon)\right)}_{\leq 1} \bigg)^{X_\lambda}
	\\
	&\leq \exp\left(-\frac{1}{\theta} \frac{{a}}{c_1 \PGW(B_X) \bar d_\lambda} \varepsilon \cdot c_1 \bar d_\lambda\right)
	\\
	&= \exp\left(- \frac{{a}}{\theta }\frac{\varepsilon}{\PGW(B_X)} \right).
	\end{align}
	This yields that
	\begin{align}
	&\EGW\left[ \left(\frac{\theta^2+\beta \theta}{\theta^2+\mathrm{e}^{\beta\theta}-1}\right)^{X_\lambda}-\left(\mathrm{e}^{-\frac{\beta}{\theta}}\left(1+\frac{\beta}{\theta}(1-\varepsilon)\right)\right)^{X_\lambda}\right]
	\\
	&\geq \left[ \frac{\varepsilon}{\PGW(B_X)} +\exp\left(-\frac{{a}}{\theta}\frac{\varepsilon}{\PGW(B_X)}\right) - \exp\left(- \frac{{a}}{\theta }\frac{\varepsilon}{\PGW(B_X)}\right) \right] \PGW(B_X)
	\\
	&=\varepsilon
	\end{align}
	and using the first part of Theorem \ref{thrm:GWloops}, we obtain the result.
\end{proof}

\begin{proof}[Proof of Theorem \ref{thrm:ex-infinite-loops}, part 1]
	We can deduce this from Corollary \ref{cor:GWloops-large-expectation} by considering the deterministic offspring distribution $X=1$ and setting $c_1:=c_2:=1$, hence $\PGW(B_X)=1$.
\end{proof}

\subsection{Absence of long loops}

We start by considering the \GW tree. The following lemma is sufficient to prove Theorem \ref{thrm:GWloops}, part 2, as it even shows exponential decay.
\begin{lem} \label{lem:GWexpdecay}
	For $\sigma^m_{n}:=\Pqu \left( E_{{T^X_{r,n}}}^{r \to m}\right)$,
	\begin{align}
	\tilde q:=\EGW\left[X \mathrm{e}^{-X\beta/\theta} \left(1+\frac{\mathrm{e}^{\beta\theta}-1}{\theta^2}\right)^{X-1} \right]\frac{\mathrm{e}^{\beta\theta}-1}{\theta^2}
	\end{align}
	and all $n, m \in \mathbb N$ with $n \geq m$ we have $\sigma^m_n \leq {\tilde q}^{m-1}$.
\end{lem}
\begin{proof} 
	Let us fix $m,n \in \mathbb N$ and write $T_0:=T^X_{r,n}$ for a realisation of the \GW tree rooted in $r$ and cut after $n$ generations. Furthermore, given such a realisation, consider the setting of (\ref{eq:tree-subtrees-setting}), in particular $d:=X_r$ is the number of children of the root $r$. Then for $ J \subseteq \{1,\ldots,d\}$ and on 
	\begin{align}
	\hat A_{J}:=\left [N_\beta^{\{r,x_j\}}\left \lbrace \begin{array}{cl} \geq 1 & \text{ if } j \in J \\  =0 & \text{ otherwise}
	\end{array} \right. \right ] 
	\end{align}
	by Proposition \ref{prop:loopnumbers-subgraphs} we have
	\begin{align}
		L_{T_0} \leq 1 + \sum_{j=1}^d L_{T_j} - 2|J| + \sum_{j\in J} N_\beta^{\{r,x_j\}}.
	\end{align}
	Furthermore, being given the event $E_{T_0}^{r\to m}\cap \hat A_{J}$ means that we can find at least one $i\in J$ such that there is a loop within $T_i$ containing $x_i$ and reaching $m-1$ generations. Therefore we have
	\begin{align}
		E_{T_0}^{r\to m} \cap \hat A_{J}
		\subseteq \bigcup_{i\in J} E_{T_i}^{x_i \to m-1} \cap \hat A_{J}.
	\end{align}
	Using Corollary \ref{cor:part-func-estimation} we obtain
	\begin{align}
		&\mathbb P_{T_0}^\theta \left( E_{T_0}^{r\to m} \cap \hat A_{J} \right) 
		\notag \\
		&\leq \sum_{i\in J} \frac{\mathbb E_{T_0} \bigg [ \mathbf{1}\left[{E_{T_i}^{x_i \to m-1}}\right] \mathbf{1}_{\hat A_{J}} \, \theta^{1 -2|J|+\sum_{j\in J} N_\beta^{\{r,x_j\}}} \prod_{j=1}^d \theta^{L_{T_j}} \bigg ]}{\theta \mathrm{e}^{-\beta d + \beta d/\theta}\prod_{j=1}^d \mathbb E_{T_j}[\theta^{L_{T_j}}]}
		\\
		&= \mathrm{e}^{-\beta d/\theta} \theta^{-2|J|} (\mathrm{e}^{\beta \theta}-1)^{|J|} \sum_{i \in J} \mathbb P_{T_i}^\theta(E_{T_i}^{x_i \to m-1}).
	\end{align}
	This yields
	\begin{align}
		&\EGW\left[\mathbb{P}_{T_0}^\theta \left( E_{{T_0}}^{r \to m}\right) \bigg| X_r=d\right]
		\\
		&= \; \sum_{\mathclap{J \subseteq \{1,\ldots,d\}}} \; \EGW \left[\mathbb P_{T_0}^\theta \left( E_{T_0}^{r\to m} \cap \hat A_J \right) \bigg| X_r=d\right]
		\\
		&\leq \; \sum_{\mathclap{J \subseteq \{1,\ldots,d\}}} \;
		\mathrm{e}^{-\beta d/\theta} \left(\frac{\mathrm{e}^{\beta\theta}-1}{\theta^2}\right)^{|J|} \sum_{i \in J} {\EGW\left[ \mathbb P_{T_i}^\theta(E_{T_i}^{x_i \to m-1})\bigg| X_r=d\right]}
		\\	
		&= \; \sum_{\mathclap{J \subseteq \{1,\ldots,d\}}} \;
		\mathrm{e}^{-\beta d/\theta} \left(\frac{\mathrm{e}^{\beta\theta}-1}{\theta^2}\right)^{|J|} |J| \sigma^{m-1}_{n-1},	
	\end{align}
	where we used self-similarity in expectation. Therefore we conclude
	\begin{align}
		\sigma^m_n &= \sum_{d\in \mathbb N_0} \PGW[X_r=d]\cdot \EGW\left[\mathbb{P}_{T_0}^\theta \left( E_{{T_0}}^{r \to m}\right) \bigg| X_r=d\right] 
		\\
		&\leq \sigma^{m-1}_{n-1} \sum_{d \in \mathbb N_0} \PGW[X_r=d] \cdot \mathrm{e}^{-\beta d/\theta} \sum_{k=0}^d {d \choose k} k \left(\frac{\mathrm{e}^{\beta\theta}-1}{\theta^2}\right)^k
		\\
		&= \sigma^{m-1}_{n-1} \sum_{d \in \mathbb N_0} \PGW[X_r=d] \cdot \mathrm{e}^{-\beta d/\theta} d \, \frac{\mathrm{e}^{\beta\theta}-1}{\theta^2} \left(1+\frac{\mathrm{e}^{\beta\theta}-1}{\theta^2}\right)^{d-1}
		\\
		&= \sigma_{n-1}^{m-1} \, \EGW\left[X \mathrm{e}^{-X\beta/\theta} \left(1+\frac{\mathrm{e}^{\beta\theta}-1}{\theta^2}\right)^{X-1} \right]\frac{\mathrm{e}^{\beta\theta}-1}{\theta^2}.
	\end{align}
\end{proof}

\begin{proof}[Proof of Corollary \ref{cor:GWloops-large-expectation}, part 2]
	By making use of the estimation
	\begin{align} \label{eq:thrm2expestimation}
	\begin{split}
	& \mathrm{e}^{-X \beta/\theta} \left(1+\frac{\mathrm{e}^{\beta \theta}-1}{\theta^2}\right)^{X-1} 
	\\
	&\leq \exp\left( -X\frac{\beta}{\theta} + X\frac{\mathrm{e}^{\beta \theta}-1}{\theta^2}\right)
	\cdot {\underbrace{\left[\exp\left(-\frac{\mathrm{e}^{\beta\theta}-1}{\theta^2}\right) \left(1+\frac{\mathrm{e}^{\beta\theta}-1}{\theta^2}\right)\right]}_{\leq 1}}^{X}
	\end{split}
	\\
	&\leq \exp\left(X\frac{\mathrm{e}^{\beta\theta}-\beta\theta-1}{\theta^2}\right)	
	\end{align}
	we see that for an offspring distribution $X$ bounded by $d \in \mathbb N$ and for $\beta \leq \frac{q}{d}$ we calculate
	\begin{align}
	&\EGW\left[X \mathrm{e}^{-X \beta/\theta} \left(1+\frac{\mathrm{e}^{\beta \theta}-1}{\theta^2}\right)^{X-1}\right] \frac{\mathrm{e}^{\beta \theta}-1}{\theta^2}
	\\
	&\leq d \cdot \exp\left(d \frac{\mathrm{e}^{\beta\theta}-\beta\theta-1}{\theta^2}\right) \frac{\mathrm{e}^{\beta \theta}-1}{\theta^2}
	\\
	& \leq \frac{d}{\theta^2}\left(\mathrm{e}^{q \theta/d}-1\right)\exp\left(\frac{d}{\theta^2} \left(\mathrm{e}^{q\theta/d}-\frac{q\theta}{d}-1\right)\right) \\
	&=:c_d . \label{eq:d-reg-subcritical-proof-2}
	\end{align}
	Since $\lim_{d\to\infty} c_d=\frac{q}{\theta}<1$ holds, we may pick $d_0 \in \mathbb N$ such that for all $d \geq d_0$ the condition in the second part of Theorem \ref{thrm:GWloops} is fulfilled.
\end{proof}

\begin{proof}[Proof of Theorem \ref{thrm:ex-infinite-loops}, part 2] This follows from the second part of Corollary \ref{cor:GWloops-large-expectation} by picking the deterministic offspring distribution $X=d$.
	
\end{proof}

%
%
\section*{Acknowledgements}

The research of BL is supported by the Alexander von Humboldt Foundation.

\end{document}